\documentclass[svgnames]{llncs}
\usepackage{stmaryrd}
\usepackage{amsmath,amssymb,amsfonts,mathrsfs,thm-restate,thmtools}

\usepackage{graphicx}
\usepackage{appendix}
\usepackage{epsf}
\usepackage{epsfig}
\usepackage{epstopdf}
\usepackage{xspace,xcolor}
\usepackage{soul}
\usepackage{latexsym}
\usepackage{tikz}
\usepackage{framed}

\usepackage{euler}
\usepackage[small,euler-digits]{eulervm}

\usepackage[backgroundcolor=Moccasin,colorinlistoftodos,obeyDraft]{todonotes}

\makeatletter
 \providecommand\@dotsep{5}
 \def\listtodoname{}
 \def\listoftodos{\@starttoc{tdo}\listtodoname}
 \makeatother

\newcounter{nmcomment}

\usepackage{caption}
\usepackage{subcaption}
\usepackage{tkz-tab}
\usetikzlibrary{arrows,positioning}
\usetikzlibrary{shapes,snakes}
\usepackage[linesnumbered,ruled,vlined]{algorithm2e}
\SetCommentSty{mycommfont}
\usepackage[linkbordercolor={0 0 1}]{hyperref}
\usepackage{nameref}

\usepackage{enumerate}
\usepackage{caption}
\usepackage{subcaption}
\usepackage[compatibility=false]{caption}
\tikzstyle{vertex}=[circle, draw, inner sep=0pt, minimum size=4.5pt]
\newcommand{\vertex}{\node[vertex]}

\newcommand{\WOH}{\ensuremath{\sf{W[1]}}-hard\xspace}
\newcommand{\NP}{\ensuremath{\sf{NP}}\xspace}

\newcommand{\NPH}{\ensuremath{\sf{NP}}-hard\xspace}

\newcommand{\fpt}{\ensuremath{\sf{FPT}}\xspace}

\newcommand{\longversion}[1]{#1}
\newcommand{\shortversion}[1]{}

\newcommand{\ff}{\textsc{Firefighting}}
\newcommand{\skv}{\textsc {Saving $k$-Vertices}}
\newcommand{\cl}{\textsc {Clique}}
\newcommand{\vc}{\textsc {Vertex Cover}}

\newcommand{\problembox}[4]{
\begin{framed}
{\sc #1} \\
\begin{tabular}{p{.18\textwidth} p{.7\textwidth}}
\hfill \bf Input: & {#2}\\
\hfill \bf Parameter: & {#3}\\
\hfill \bf Question: & {#4}\\

\end{tabular}
\end{framed}
}

 \title{On Structural Parameterizations of Firefighting}
 \author{Bireswar Das, Murali Krishna Enduri\thanks{Supported by Tata 
 Consultancy Services (TCS) research fellowship.}, Neeldhara Misra\thanks{Supported by a DST-INSPIRE Fellowship.} and I. Vinod Reddy}
 \institute{IIT Gandhinagar, India \\
 \email{\{bireswar,endurimuralikrishna,neeldhara.m,reddy\_vinod\}@iitgn.ac.in}
 }

\begin{document}
\pagestyle{plain}

 \maketitle
 \begin{abstract}
The \ff{} problem is defined as follows. At time $t=0$, a fire breaks out at a vertex of a graph. At each time step $t \geq 0$, a firefighter permanently defends (protects) an unburned vertex, and the fire then spread to all undefended neighbors from the vertices on fire. This process stops when the fire cannot spread anymore. The goal is to find a sequence of vertices for the firefighter that maximizes the number of saved (non burned) vertices.

The \ff{} problem turns out to be \NPH{} even when restricted to bipartite graphs or trees of maximum degree three. We study the parameterized complexity of the \ff{} problem for various structural parameterizations. All our parameters measure the distance to a graph class (in terms of vertex deletion) on which the \ff{} problem admits a polynomial time algorithm. Specifically, for a graph class $\mathcal{F}$ and a graph $G$, a vertex subset $S$ is called a \textit{modulator} to $\mathcal{F}$ if $G \setminus S$ belongs to $\mathcal{F}$. The parameters we consider are the sizes of modulators to graph classes such as threshold graphs, bounded diameter graphs, disjoint unions of stars, and split graphs. 

To begin with, we show that the problem is \WOH{} when parameterized by the size of a modulator to diameter at most two graphs and split graphs. In contrast to the above intractability results, we show that \ff{} is fixed parameter tractable (\fpt{}) when parameterized by the size of a modulator to threshold graphs and disjoint unions of stars, which are subclasses of diameter at most two graphs. We further investigate the kernelization complexity of these problems to find that \ff{} admits a polynomial kernel when parameterized by the size of a modulator to a clique, while it is unlikely to admit a polynomial kernel when parameterized by the size of a modulator to a disjoint union of stars. 
%
 \end{abstract}

\section{Introduction}
The \ff{} problem was introduced by Hartnell \cite{hartnell1995firefighter} to model the spread of diseases and computer viruses. 
It is a turn-based game between two players (the ``fire'' and the ``firefighter''), which is played on a graph $G$ as follows. Initially, at time t=0, a fire starts at a vertex $s$, at each following time step the
following happens. A firefighter defends one vertex which is not on fire, and the fire then spreads from each burning vertex to all its undefended
neighbors. Once a vertex is defended it remains so for all time intervals. The process stops when the fire can no longer 
spread. The natural algorithmic question associated with this game is to find a strategy that optimizes some desirable criteria, 
for instance, maximizing the number of saved vertices~\cite{cai2008firefighting}, minimizing the number of rounds, 
the number of firefighters per round~\cite{chalermsook2010resource}, or the number of burned vertices~\cite{finbow2000minimizing,cai2008firefighting}, and so on. 
These questions are well-studied in the literature, and while most variants are \NPH{}, approximation and parameterized 
algorithms have been proposed for various scenarios.  In this work, we will focus on the goal of finding a sequence of defending vertices that maximizes the number of saved (not burned) vertices and we refer to this as the \ff{} problem. We also use \skv{} to refer to the decision version of this problem, where we are given a demand~$k$ and the goal is to save at least $k$ vertices.

We  study the parameterized complexity of \ff{} with respect to various structural parameters. In particular, our focus is on distance-to-triviality parameterizations, wherein we identify classes of graphs on which the \ff{} problem is solvable in polynomial time, and understand the parameterized complexity of the problem parameterized by the distance of a graph to these graph classes. In this paper, our notion of distance to a graph class in the vertex deletion distance. More precisely, for a class $\mathcal{F}$ of graphs, we say that $X$ is an $\mathcal{F}$-modulator of a graph $G$ if there is a  subset $X\subseteq V(G)$ such that  $G \setminus X \in \mathcal{F}$.  If the size of a smallest modulator to $\mathcal{F}$ is $k$, we also say that the distance of $G$ to the class $\mathcal{F}$ is $k$. Throughout this paper, we will assume that a modulator is given to us as a part of the input. This assumption is without loss of generality since such modulators can be computed in \fpt{} time. We are now ready to describe our results. 


\paragraph{Our Contributions.} 
The \ff{} problem is \fpt{} when parameterized by the vertex cover and distance to a clique parameterizations. On the other hand, it is para-NP-hard when parameterized by feedback vertex set, tree-width and clique-width \cite{finbow2007firefighter}. However, the parameter vertex cover is very restrictive and significantly large for dense graphs. This motivates us to consider parameters that are intermediate between 
vertex cover and clique-width. In this spirit, Ganian \cite{ganian2015improving} studied the parameterized complexity of \ff{} problem for parameter twin-cover which is a generalization of vertex cover and showed that \ff{} is $\fpt$ with respect to twin-cover. Recently Chleb{\'\i}kov{\'a} et al.~\cite{chlebikova2017firefighter} showed that the problem is $\fpt$ parameterized by distance to cluster graphs which is a generalization of twin-cover.

We study the parameterized complexity of the \ff{} problem with respect to the distance from following graph classes: threshold graphs, disjoint union of stars, disjoint union of graphs of diameter at most two, and split graphs.  
Studying the parameterized complexity of \ff{} with respect to these parameters improves the understanding of the boundary between tractable and intractable parameterizations. For instance, the parameterization by distance to cluster graphs (as studied by~\cite{chlebikova2017firefighter}) directly generalizes both vertex cover and distance to clique. Observe that cluster graphs are precisely the graphs whose connected components have diameter one, and a natural generalization to consider is the class of graphs whose connected components have diameter two. Here, we show that there is a transition in complexity: the problem becomes \WOH{}. As a natural intermediate problem, we consider the subclass of graphs where every connected component is a star. Here, using ideas similar to the ones that lead to the \fpt{} algorithm for the distance to cluster parameter, we obtain a \fpt{} algorithm. The case analysis here is more delicate because we have to distinguish between the central vertex and the leaves. 

On the other hand, the distance to threshold graphs parameter directly generalizes the distance to clique parameter, while the distance to stars parameter is a generalization of vertex cover. For both of these parameters, we establish that \ff{} is \fpt{}. These being smaller parameters, our results improve several known algorithms. Note that the next ``natural'' parameter to consider after distance to stars hierarchy is the feedback vertex number, or the distance to forests; however here the problem is already \NPH{} on trees, leading to para-\NPH{}ness. Similarly, a natural next step from distance to threshold graphs is the distance to split graphs, but here also we demonstrate \WOH{}ness. Finally, a promising generalization from the distance to cluster graphs is the distance to cographs, and here we leave the parameterized complexity of the problem open. 

\begin{figure}[t]
\centering
\includegraphics[scale=0.6,trim={2.5cm 13cm 2cm 3.5cm},clip]{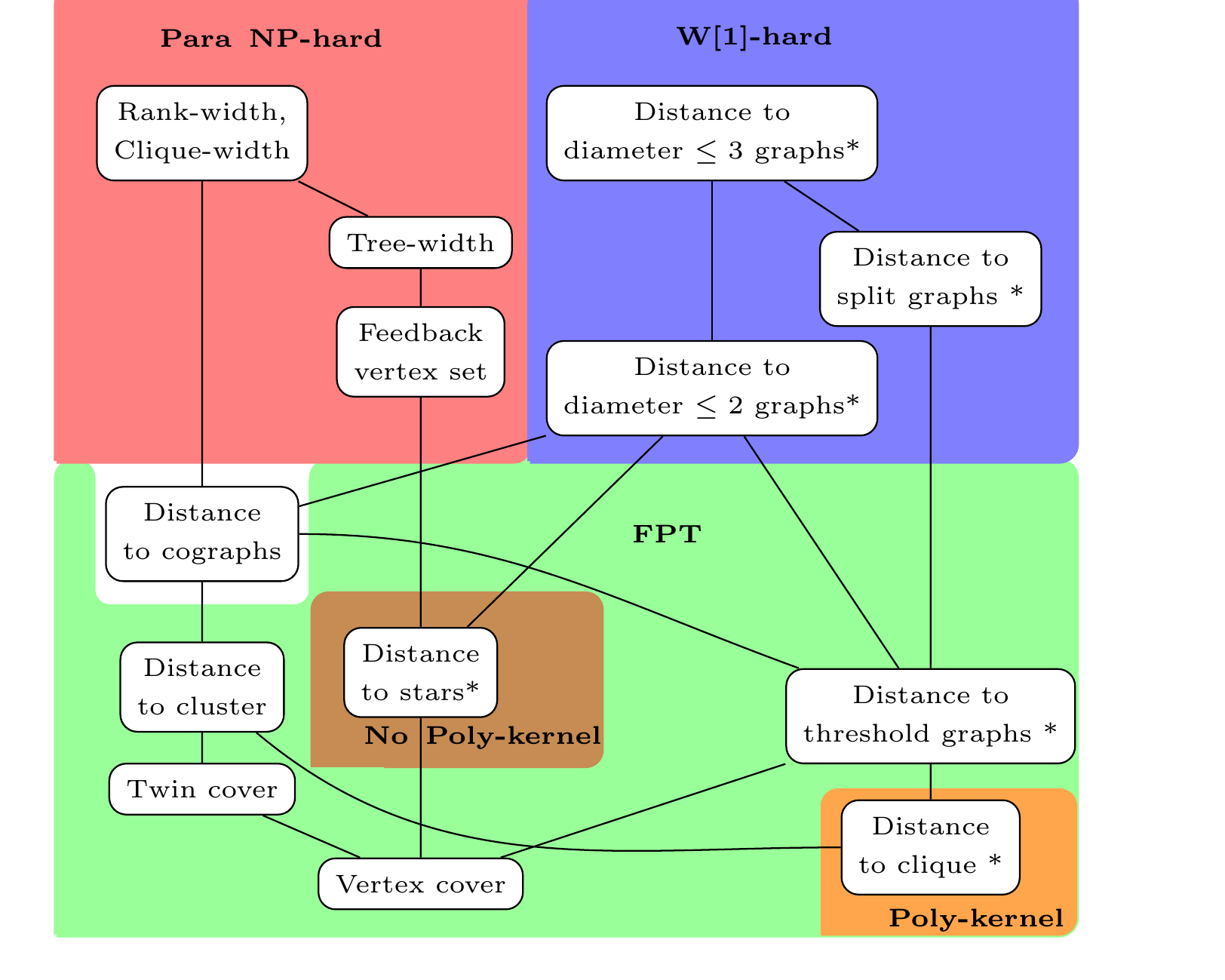}
\caption{ \hspace{0.1cm} A schematic showing the parameterized complexity of \ff{} with respect to various structural graph parameters.
There is a line between two parameters
if the parameter below is larger than the parameter
above. Results shown in this paper are marked by an asterisk (*). }
\label{fig1}
\end{figure}

We also consider the kernelization complexity of the problem and make the following advances: for the distance to clique parameterization, we demonstrate a quadratic kernel, while for the distance to stars parameterization, we show that a polynomial kernel is unlikely under standard complexity-theoretic assumptions. The kernelization complexity of the problem relative to vertex cover, however, remains an interesting open problem. We summarize our results below.

\begin{itemize}
\item We show that the problem is  fixed parameter tractable (\fpt{}) when parameterized by the size of a modulator to threshold graphs, cluster graphs and disjoint unions of stars. 
\item 	We further investigate the kernelization complexity of these problems to find that Firefighting admits a polynomial kernel when parameterized by the size of a modulator to a clique, while it is unlikely to admit a polynomial kernel when parameterized by the size of a modulator to a disjoint union of stars.
\item  Finally, in contrast to the tractability results, we show that \ff{} is \WOH{} when parameterized by the distance to split graphs. In fact, the problem remains \WOH{} with respect to the combined parameter involving the size of the modulator and the number of vertices to be saved. 
\end{itemize}

\paragraph{Methodology.} By and large, we use a standard approach for the \fpt{} algorithms: we guess the behavior of the solution on the modulator and attempt to find a solution consistent with the guessed behavior. The second part relies on exploiting the structural properties of $G \setminus X$, which is the part of the graph outside the modulator. Usually one is able to group the vertices of $G \setminus X$ based on the structure of their neighborhoods in the modulator, and argue that all vertices of the same ``type'' have a similar behavior, which leads to a controlled search space. In the case of threshold graphs, we are able to prove that simple greedy techniques work within a particular type. On the other hand, for disjoint unions of stars, we have to account for several scenarios, and the classification of $G \setminus X$ is more intricate, and accordingly, we have to account for more cases in the analysis. The hardness results follow from reductions using standard techniques, while the kernelization algorithm uses the fact that when $G \setminus X$ is a large clique, several vertices behave in a similar fashion, and this observation allows us to replace the large clique with a much smaller one --- a careful argument is required, however, to demonstrate that the instance we constructed in this fashion is indeed equivalent to the original. 

\paragraph{Related Work.} The \ff{} problem is known to be $\NP$-hard even for special classes of graphs, including bipartite graphs \cite{macgillivray2003firefighter}, trees of maximum degree three \cite{finbow2007firefighter}
and cubic graphs~\cite{king2010firefighter}. The firefighter problem can be solved in linear time on split graphs and co-graphs \cite{fomin2016firefighter}. From the parameterized complexity point of view, the firefighting problem is $W[1]$-hard when parameterized by the number of saved vertices for bipartite graphs \cite{bazgan2014parameterized}. Cai et al. \cite{cai2008firefighting} propose several \fpt{} algorithms and polynomial kernels, and they consider the following parameters: the number of saved vertices, the number of saved leaves, and the number of protected vertices. Leung \cite{leung2011fixed} use the random separation method to give \fpt{} algorithms on general graphs parameterized by the number of burnt vertices and on degree bounded graphs and unicyclic graphs parameterized by the number of protected vertices. We refer the reader to the survey~\cite{Finbow:2009vma}, 
as well as the references within, for more details. 

\section{Preliminaries}
In this section, we introduce the notation and the terminology that we will need to describe our algorithms. Most of our notation is standard. We use $[k]$ to denote the set $\{1,2,\ldots,k\}$. All graphs we consider in this paper are undirected, connected, finite and simple.
For a  graph $G=(V,E)$, let $V(G)$ and $E(G)$ denote the vertex set and edge set of $G$ respectively.
An edge in $E$ between vertices $x$ and $y$ is denoted as $xy$ for simplicity. 
For a  subset $X \subseteq V(G)$, the graph $G[X]$ denotes the subgraph of $G$ induced by vertices of $X$. Also, we abuse notation and use $G \setminus X$ to refer to the graph obtained from $G$ after removing the vertex set $X$. 
For a vertex $v\in V(G)$,
$N(v)$ denotes the set of vertices adjacent to $v$ and $N[v] = N(v) \cup \{v\}$ is the closed neighborhood of $v$. \shortversion{We use standard notions from parameterized complexity, following~\cite{cygan2011parameterized}. We refer the reader to the Appendix for the most relevant definitions required in this paper, and to~\cite{cygan2011parameterized} for a more detailed introduction.}

\paragraph{Graph Classes.} We now define the graph classes that we will encounter frequently. 

\begin{itemize}
\item 	A graph is a \emph {split graph} if its vertices can be partitioned into a clique and an independent set.  Split graphs are $P_5$-free \cite{foldes1976split}.
\item The class of $P_4$-free graphs are called \emph{co-graphs}. 
\item A graph is a \emph{threshold graph} if it can be constructed recursively by adding an isolated vertex or a universal vertex.
\item A cluster graph is a disjoint union of complete graphs. Cluster graphs are also $P_3$-free graphs.

\end{itemize}

 It is easy to see that a graph that is both split and co-graph is a threshold graph. We denote  threshold graph (or a split graph) with $G = (C,I)$ where $C$ and $I$ denote a partition
of $G$ into a clique and an independent set. For any two vertices $x, y$ in a threshold graph $G$ we have either $N(x)\subseteq N[y]$ or $N(y)\subseteq N[x]$. For a class of graphs $\mathcal{F}$, the distance to $\mathcal{F}$ of a graph
$G$ is the minimum number of vertices to be deleted from $G$ to get a graph in $\mathcal{F}$. 


\longversion{\paragraph{Parameterized Complexity.}
A parameterized problem is a pair $Q\subseteq \Sigma^*\times \mathbb{N}$, where $\Sigma$ is fixed finite alphabet. For an instance $(x, k) \in \Sigma^*\times \mathbb{N}$, $k$ is called the
parameter. We say that 
a parameterized problem $Q$ is {\it fixed parameter tractable} ($\fpt$) if there exists an algorithm  and a computable function $f : \mathbb{N} \rightarrow \mathbb{N}$ such that given $(x,k) \in \Sigma^*\times \mathbb{N}$ the algorithm correctly decides  whether $(x,k) \in Q$ in
$f(k)· |(x,k)|^{O(1)}$ time. The function $f$ is usually superpolynomial and only depends on parameter $k$. The class $XP$ contains the problems  which are solvable in time $|(x,k)|^{f(k)}$, the exponent of the running time depends on the parameter $k$. A problem is {\it para-NP-hard} if it is $\NP$-hard for some fixed values of the parameter. The complexity class of parameterized intractability is called $W[1]$ (see \cite{downey1999parameterized} for the definition). A {\it kernelization } algorithm is a polynomial time algorithm that takes an instance $(x, k)$ of a parameterized
problem $Q$ as input and outputs an equivalent instance $(x', k')$ of $P$ such that $|x'|\leq h(k)$ for some computable function $h$ and $k'\leq k$. If
$h$ is polynomial then we say that $(x', k')$ is a polynomial kernel. Let $P$ and $Q$ be two parameterized problems. A parameterized reduction from $P$ to $Q$ is a mapping 
$g: \Sigma^* \rightarrow  \Sigma^*$ such that (i) for all $x \in  \Sigma^*$ we have $x$ is a yes instance of $P$ $\Leftrightarrow$ $g(x)$ 
is a yes instance of $Q$.
(ii) $g$ can be computed in $f(k) |x| ^{O(1)}$ time, where $f$ is computable function and $k$ is the parameter of $x$.
(iii) $k' \leq h(k)$ for some computable function $h$, where $k$ and $k'$ are parameters of $x$ and $g(x)$ respectively.  
For more details on parameterized complexity the reader is referred to  \cite{downey1999parameterized,cygan2015parameterized}.}

\longversion{We now briefly justify our assumption about the modulators being given as a part of the input to our problems. Consider the following result. 

\begin{lemma}\cite{cai1996fixed} \label{lem-recognize}
Let $\mathcal{F}$ be a graph class characterized by the finite forbidden induced subgraphs $H_1,\cdots,H_l$.
Given a graph $G$ and integer $k$, there is an $\fpt$ algorithm that finds a subset $X \subseteq V(G)$ of size at most $k$ such that $G \setminus X \in \mathcal{F}$ in $O(d^k n^d)$, where $d$ is the size of largest forbidden subgraph.
\end{lemma}

The class of cluster graphs, threshold graphs and split graphs are $P_3$, $P_4$ are $P_5$ free respectively. By using above lemma, Given a graph $G$ and integer $k$, the problem of deciding whether there exists a set $X$ of vertices of size at most $k$ whose deletion results in a cluster graph, threshold graph and split graph is fixed parameter tractable. }

\paragraph{The Firefighting Problem.} Finally, we formally define the problems that we consider in this paper, where $\mathcal{F}$ is a class of graphs or a graph property. 

\problembox{\ff{}[$\mathcal{F}$]}{A graph $G$, a vertex $s$, a modulator $X \subseteq V(G)$ such that $G \setminus X \in \mathcal{F}$.}{The size $k := |X|$ of the modulator to $\mathcal{F}$.}{Find a strategy that maximizes the number of saved vertices when a fire starts at $s$?}

\problembox{\skv{}[$\mathcal{F}$]}{A graph $G$, a vertex $s$, a modulator $X \subseteq V(G)$ such that $G \setminus X \in \mathcal{F}$, and an integer $k$.}{The size $\ell := |X|$ of the modulator to $\mathcal{F}$.}{Does there exists a strategy that saves at least $k$ vertices when a fire starts at $s$?}

Whenever $\mathcal{F}$ is clear from the context, we drop the explicit mention of it in the name of the problem. We also abuse notation and use $k$ differently in the two definitions, to retain consistency with standard notation when we present reductions in the context of \skv{}. 
\section{The Parameterized Complexity of Firefighting}\label{sec-basics}

Let $(G,s)$ be an instance of \ff{} problem.
The vertices that have not been burned by the fire at the
end of the process are called \emph{saved} (including defended vertices).
A vertex is \emph{burned} if it is on fire. A strategy for \ff{} instance $(G,s)$ is a sequence $S$ of vertices $\{v_1, \cdots, v_l\}$ 
where $v_i$ represents the position of the firefighter at time step $i$.  
We say that sequence $S$ is a \emph{valid strategy} for $(G,s)$ 
if vertex $v_i$ is not burning at the start of time step $i$ and the process stops at time
step $l$. A strategy $S$ is called \emph{minimal} if no subset of $S$ yields a strategy that saves same number of vertices as $S$.  
A strategy is \emph{optimal} if it is minimal and saves maximum number of vertices. 

The following results from the literature will be useful for our algorithms. 
\begin{lemma}\label{lem-stcheck} \cite{fomin2016firefighter}
Let $G$ be a graph and $s$ be a vertex of $G$. Given an ordered set $S$ of vertices
of $G$, we can verify whether $S$ is a valid strategy for the \ff{} problem on $(G,s)$ 
and count the number of vertices saved by S in $O(n + m)$ time, where $n$ and $m$ denote the number of vertices and edges in $G$ respectively. 
\end{lemma}
\begin{proof}
Let $S =\{v_1,\cdots,v_k\}$ is strategy in this order. 
 To verify whether the sequence $S$ is a valid strategy, we do
 BFS on the graph $G$ starting from $s$. Find the distance $d(s,v)$ from the source $s$ to each vertex $v$ of $S$ on graph 
 $G[(V(G) \setminus S) \cup \{v\}]$. For $i = 1$ to $k$,
If $S$ is a valid strategy then distance from source $s$ to the vertex $v_i$,  $d(s, v_i ) \geq i$, otherwise 
vertex $v_i$ will be burned before time step $i$. Observe that the number of vertices burned by $S$ in $G$ equal to
the number of vertices reachable from $s$ in $G \setminus S$, which can be found by applying BFS on the graph $G \setminus S$ starting from $s$.\qed
\end{proof}

\begin{corollary}\label{coro-burned}
 Let $S$ be an optimal strategy for the \ff{} problem on $(G, s)$ then the number of vertices burned by $S$ in $G$ is equal to
the number of vertices reachable from $s$ in $G \setminus S$.
\end{corollary}
%

\begin{lemma}\label{lem-optlength1}\cite{fomin2016firefighter}
 Let $(G,s)$ be an instance of the \ff{} problem, and let $l$ be the length of a longest
induced path in $G$ starting from $s$. Then any optimal strategy can defend at most $l$ vertices.
\end{lemma}
\begin{proof}
 Let $S=\{v_1,\cdots,v_t\}$ be an optimal strategy defended in this order. 
 Since $S$ is a valid strategy there is an induced path $P$ from $s$ to $v_t$ such that all the vertices on $P$ are burned
 except $v_t$. Let $P$ be a shortest such path, then $P$ contains at least $t+1$ vertices: otherwise $v_t$ will be burned before
 time $t$. Since the length of longest induced path in $G$ starting from $s$ is $l$, we have $t+1\leq l+1$ which implies $t \leq l$. \qed
\end{proof}

\begin{lemma} \cite{fomin2016firefighter}  \label{lem-xp}
The \ff{} problem can be solved in $O(n^l)$ time on graphs with length of longest induced path is at most $l-1$.
\end{lemma}
\begin{proof}
Let $(G,s)$ be an instance of firefighting problem. 
 Since the length of the longest induced path in $G$ is at most $l$. 
 From Lemma~\ref{lem-optlength1} any optimal strategy can defend at most $l$ vertices in $G$. 
 We list out all possible subsets $S$ of $V(G)$ of size at most $l$ in $O(n^l)$ time. For each such subset $S$ using Lemma~\ref{lem-stcheck}
 test whether $S$ is valid strategy and count the number of vertices saved by $S$. 
 The optimal strategy is the one which saves maximum number of vertices. \qed
\end{proof}

\subsection{Parameterization by Distance to Threshold Graphs}

In this section we give an \fpt algorithm for the \ff{} problem parameterized by the distance to threshold graphs. Without loss of generality, we assume that threshold graph is connected, otherwise all the connected components that do not contain the source vertex are trivially saved. 
\begin{corollary}\shortversion{[$\star$]}\label{coro-length}
 Let $(G,s)$ be an instance of \ff{} problem. 
 Then any optimal strategy can defend at most $2k+2$ vertices,
where $k$ is the distance to threshold graphs. 
\end{corollary}
\longversion{
\begin{proof}
 We know that fire always spreads along an induced path in $G$. As the length of the longest induced path in $G$ is at most $2k+2$,
 from Lemma~\ref{lem-optlength1} 
any optimal strategy can defend at most $2k+2$ vertices in $G$. \qed 
\end{proof}
}
Let $G$ be a graph and $X \subseteq V(G)$ of size $k$ such that $G \setminus X=(C,I)$ is a threshold graph.
We partition the vertices of clique $C$ and independent set $I$ in $G \setminus X$ based on their neighborhoods in $X$. In particular, for every subset $Y \subseteq X$, let:
$T_Y^C:=\{x \in C ~|~ N(x) \cap X=Y\} \mbox{ and } T_Y^I:=\{x \in I ~|~ N(x) \cap X=Y\}.$ 

Notice that in this way we can partition vertices of $G \setminus X$ into at most $2^{k+1}$ subsets (called types), two for each $Y \subseteq X$. Observe that all vertices in a type have same neighbors in $X$, where as they may have different neighbors inside the threshold graph. The following result shows that when we need to choose a defending vertex from a type the best strategy is to defend a highest degree vertex in that type. For a strategy $S$, let $sav(S)$ denote the number of vertices saved by the strategy $S$.
\begin{lemma}\shortversion{[$\star$]}\label{optchoose}
Let $v_1, v_2$ be two vertices in a type $T$ such that $N(v_2) \subseteq N[v_1]$. Let $S$ be a 
strategy containing $v_2$, which is defended at time step $i$. If $v_1 \notin S$ and not burning at the start of time step $i$ then $sav(S) \leq sav(S')$ where $S'$ is obtained from $S$ by replacing $v_2$ with $v_1$ at time step $i$.
\end{lemma}
\longversion{\begin{proof}
Using Corollary~\ref{coro-burned}, the vertices burned by $S$ in $G$ are the vertices which are reachable from $s$ in $G \setminus S$. We show that every vertex $u$ which is reachable from $s$ in $G \setminus S'$ is also reachable from   $s$ in $G \setminus S$. Let $P$ be a path between $s$ and $u$ in $G \setminus S'$. If $v_2 \notin P$ then $P$ is a path in $G \setminus S$.
If $v_2 \in P$, then using the facts that $v_1$ is a burned vertex in strategy $S$ and $N(v_2) \subseteq N[v_1]$, 
we can see that $P'=P \setminus \{v_2\} \cup\{v_1\}$ is
a path between $s$ and $u$ in $G \setminus S$. Therefore the number of vertices burned by $S$ is at least the number of  vertices burned by $S'$, which implies  $sav(S) \leq sav(S')$. 
\end{proof}
}
 Our \fpt{} algorithm now follows by guessing, for each step in the defending sequence, if the vertex is from the modulator or the type of the vertex from $G \setminus X$. Then it simulates the sequence (by substituting for each guess of a type, a greedily chosen vertex from that type) to check if it is a valid solution.  
 
\begin{theorem}\shortversion{[$\star$]}
The \ff{} problem can be solved in $O((2^{k+1}+k)^{2k+2} (n+m))$ time when parameterized by size of modulator to threshold graphs. 
\end{theorem} 
\longversion{ 
\begin{proof}
Partition the vertices of $G \setminus X$ into at most $2^{k+1}$ sets. Each time when we want to defend a vertex we only choose from $2^{k+1}$ (types) $+k$ (size of $X$). From Lemma~\ref{optchoose} we know which vertex has to be defended in a given type. 

From Corollary~\ref{coro-length} it is clear that we only need to defend at most $2k+2$ times, therefore there are at most $(2^{k+1}+k)^{2k+2}$ possible
firefighting strategies. For each such strategy $S$, 
using Lemma~\ref{lem-stcheck}, test whether $S$ is valid strategy and count the number of vertices saved by $S$. The strategy which saves maximum number of vertices is the optimal strategy. Therefore this procedure takes $O((2^{k+1}+k)^{2k+2}(n+m))$ time.\qed

\end{proof}
}

\subsection{Parameterization by Distance to Stars}\label{sec-stars}
In this section we design an $\fpt$ algorithm for the \ff{} problem parameterized by the distance to stars. 
Recall that this is the minimum number of vertices to be deleted from $G$ to get a disjoint union of stars.
Let $X$ be a $k$-sized modulator to disjoint union of stars. Our first observation follows easily from the bound on the length of the longest induced path. 


\begin{corollary}\label{coro-optlength-stars}
 Let $(G,s)$ be an instance of \ff{} problem.  Any optimal strategy can defend at most $4k+2$ vertices, 
 where $k$ is distance to disjoint union of stars.
\end{corollary}
\begin{proof}
 We know that fire always spreads along an induced path in $G$. As the length of maximum induced path in $G$ is at most $4k+2$,
 from Lemma~\ref{lem-optlength1} any optimal strategy can defend at most $4k+2$ vertices in $G$. \qed 
\end{proof}
We now define a notion of equivalent stars, which will lead us to partitioning of the stars in $G \setminus X$ into types as before. For a star $\mathbb{S}$ with center $c$ in $G \setminus X$, we use $B(\mathbb{S})$ to denote the set of vertices in $\mathbb{S}$ that have a neighbor in $X$, and call these the \textit{border vertices}. Further, for a nonempty subset $Y \subseteq X$, we use $B_Y(\mathbb{S})$ to denote the set of vertices in $\mathbb{S}$ whose neighborhood in $X$ is exactly $Y$.

\begin{definition}\label{def-equivalent-stars}
Let $X \subseteq V(G)$ such that $G \setminus X$ is a disjoint union of stars.
We call two stars $\mathbb{S}_i$ and $\mathbb{S}_j$ are equivalent if, (a) $N(c_i)=N(c_j)$, where $c_i$ and $c_j$ are the centers of stars $\mathbb{S}_i$ and $\mathbb{S}_j$ respectively. (b) $N(\mathbb{S}_i)=N(\mathbb{S}_j)$, (c) $|B(\mathbb{S}_i)|=|B(\mathbb{S}_j)|$
and (d) For every non-empty subset $Y \subseteq X$, we have that $|B_Y(\mathbb{S}_i)|=|B_Y(\mathbb{S}_j)|$.
\end{definition}
For an equivalence class $T$, we use $b_T$ to denote the size of the border for any star in $T$. Our next result bounds the number of equivalence classes. The bound follows roughly from the fact that one can associate a signature with an equivalence class based on condition (c) in Definition~\ref{def-equivalent-stars}, which, in turn, can be put in one-one correspondence with strings of length $2^k$ over an alphabet of size $\ell$, where $\ell$ is the maximum possible value of $|B(\mathbb{S})|$ in $G \setminus X$.

\begin{lemma}\shortversion{[$\star$]}\label{lem-equipartition-stars}
Let $\ell$ be defined as above and let $X \subseteq V(G)$ be a modulator to disjoint union of stars of size $k$. Then, the stars of $G \setminus X$  can be partitioned into at most~$O(2^{2k} \ell^{2^k})$ equivalence classes.
\end{lemma}
\longversion{
\begin{proof}
First, partition the stars in $G \setminus X$ into at most $2^{2k}$
sets such that all stars in each set satisfies conditions (a) and (b) of Definition~\ref{def-equivalent-stars}. 
Now each set of the partition can be further divided based on the value of $|B(\mathbb{S}_i)|$ for each star $\mathbb{S}_i$ in that set.
As $1 \leq |B(\mathbb{S}_i)| \leq \ell$ for all $\mathbb{S}_i \in G \setminus X$, each set can be partitioned into at most $\ell$ sets.  
In order to satisfy condition (d) in Definition~\ref{def-equivalent-stars} each set further partitioned into $\ell^{2^k-1}$ sets. Combining all, there are at most $O(2^{2k} \ell^{2^k})$ equivalent star partitions of stars in $G \setminus X$. \qed

\end{proof}
}

Our \fpt{} algorithm begins by guessing the behavior of the solution on the modulator, and builds on the fact that there are a bounded number of equivalence classes. We defer the details of this algorithm to a full version because of space constraints, and note that it is similar --- in spirit --- to the approach used for the cluster vertex deletion parameter in~\cite{chlebikova2017firefighter}.

\begin{theorem}\label{th-ffcvd}
The \ff{} problem is \fpt{} when parameterized by the distance from the class of the disjoint union of stars.
\end{theorem}  
\begin{proof}
Without loss of generality, assume $s \in X$: If $s \notin X$ then $G \setminus (X \cup \{s\})$ is also a disjoint union of stars.
As a first step,
we guess the defended $(X_d)$, saved $(X_s)$ (which are not in $X_d$) and burned $(X_b)$ vertices from $X$ in $O(3^k)$ time.
Since the length of any optimal defending sequence is at most $4k+2$, we can guess the set of defended vertices 
$(X_d)$ from $X$ along with their positions in an optimal defending sequence in $O(k^k(4k+2)^k)$ time. The rest of this proof is divided into two cases depending on whether $X_s=\emptyset$ or $X_s\neq \emptyset$.

\paragraph{Case 1:  $X_s = \emptyset$.} Since the fire starts at $s$, all the vertices in $V(G) \setminus X_d$ which are not reachable from $s$ in $G[V(G) \setminus X_d]$ are saved from fire. We begin by deleting them from $G$ and concentrate on the connected component $C_s$ of $G \setminus X_d$ containing the vertex $s$. In particular, the  instance of firefighting problem that we solve from here is $(G[C_s],s,X_b)$ where $X_b$ is the modulator to disjoint union to stars. Also, we note at this point that if $X_b = \emptyset$, then we are done by a straightforward case analysis on the nature of $N(s)$ and simulating the guessed defence sequence. Therefore, we assume that $X_b \neq \emptyset$.

Partition the stars in $C_s \setminus X_b$ into equivalence classes (with respect to $X_b$) using Definition~\ref{def-equivalent-stars}.
We divide equivalence classes of stars in $C_s \setminus X_b$ into two categories: equivalence classes with $b_T \leq 4k+2$ as one category and remaining as the other. We combine all equivalence classes $T$, with $b_T>4k+2$ to into one new equivalence class $T_{new}$.
Therefore from Lemma~\ref{lem-equipartition-stars} there are at most $O(k2^k {(4k+2)}^{2^k})$ equivalence classes of stars in $C_s \setminus X_b$. From Corollary~\ref{coro-optlength-stars}, the length of any optimal defending sequence $S$ is at most $4k+2$, therefore
we can defend for at most $4k+2$ time steps. Since we have already guessed defending vertices $X_d$ from $X$, the remaining defending
vertices in $S$ has to be from $C_s \setminus X_b$. 
At each time step when we want to choose a defending vertex from $C_s \setminus X_b$, we choose from $O(k2^k {(4k+2)}^{2^k})$ equivalence classes. Therefore there are at most $O({(k2^k {(4k+2)}^{2^k}))}^{4k+2}$ possible firefighting defending sequences. Let $S$ be one such possible defending sequence, then each position of $S$ represents either a vertex of $X_d$ or an equivalence class in $G \setminus X_b$.

Now we describe the procedure to select a defending vertex from an equivalence class. 
Let $T$ be an equivalence class such that $b_T \leq 4k+2$.
First, count $q$, the number of times $T$ appeared in strategy $S$.
Since $T$ can appear at most $4k+2$ times, $q \leq 4k+2$.
In each star, we need to defend at least one vertex so we can defend at most $4k+2$ stars.
Order the stars in $T$ in descending order of size $|\mathbb{S}_i \setminus B(\mathbb{S}_i)|$.
We choose defending vertices from set $D=\cup_{i=1}^{4k+2} (B(\mathbb{S}_i) \cup \{c_i\})$ (i.e, union is over first $4k+2$ stars ordered according to $|\mathbb{S}_i \setminus B(\mathbb{S}_i)|$). The size of $D$ is at most $(4k+2)(4k+3)$.

In equivalence class $T_{new}$, we only need to defend center of stars in $T_{new}$. 
First count $q$ number of times $T_{new}$ is appeared in sequence $S$, we can only defend centers of $q$ different stars. 
Since each star center is at different distance from source $s$, so we can not greedily defend center of a star for which 
$|\mathbb{S}_i \setminus B(\mathbb{S}_i)|$ is maximum. The defending vertex set $D$ is defined as follows. For each $i \in [4k+2]$, 
find top $i$ stars according to size of $|\mathbb{S}_i \setminus B(\mathbb{S}_i)|$ whose center is at a distance $i$ from source $s$. 
Add the centers of these stars to set $D$. It is easy to see that $|D| \leq (4k+2)^2$.

\paragraph{Case 2:  $X_s \neq \emptyset$.} Let $X^s_1,X^s_2,\cdots,X^s_{m'}$ and $X^b_1,X^b_2,\cdots,X^b_{m''}$ be the connected components of $G[X_s]$ and $G[X_b]$ respectively. Recall that we are working in the connected component containing $s$ in the graph $G \setminus X_d$. Let $\mathbb{S}$ be a star in $G \setminus (X_b \cup X_s)$. We say that $\mathbb{S}$ is \textit{vulnerable} if it has at least one vertex $u$ that has a neighbor in $X_b$ and at least one vertex $v$ that has a neighbor in $X_s$ (potentially $u=v$). For any vulnerable star $\mathbb{S}$, note that any firefighting strategy that respects the semantics of $X_s, X_b$ and $X_d$ must defend at least one vertex in $\mathbb{S}$.
Indeed, suppose not, and let $S$ be a firefighting strategy that violates this criteria. Then after $S$ is executed, $\mathbb{S}$ has at least one undefended vertex (say $p$) that is a neighbor of some vertex in $X_b$ (say $x$) and one undefended vertex (say $q$) that is a neighbor of some vertex in $X_d$ (say $y$). Then, note that the $(x,p,q,y)$ is a path from $X_b$ to $X_s$ in the graph after the $S$ is executed, contradicting the semantics of $X_s, X_b$ and $X_d$. Observe that this argument works even if $p=q$. To discuss the consequence of this observation further, we employ the following notation:

$$B^b(\mathbb{S}_i):=\{x \in \mathbb{S}_i ~|~ N(x) \cap X_b \neq \emptyset\} \mbox{ and }B^s(\mathbb{S}_i):=\{x \in \mathbb{S}_i ~|~ N(x) \cap X_s \neq \emptyset\}.$$

Note now that since at least one vertex has to be defended from every vulnerable star, the total number of vulnerable stars $\ell$ is at most $4k+2$. 


We place all the vulnerable stars in one equivalence class (call this $T^\prime$). Now two cases remain: the first is that every border vertex of the star $\mathbb{S}$ has all its neighbors in $X_s$ or $B^b(\mathbb{S}_i) = \emptyset$. We place all such stars also in one equivalence class (call this $T^\star$).  Observe that in this case, representative vertices for the placement of firefighters may be chosen arbitrarily. This is because the entire star is already safe with respect to any firefighting strategy that respects the semantics of $X_b, X_s$ and $X_d$. Now consider the stars that have all their neighbors in $X_b$, and note that it is not possible to save vertices in $X_s$ by defending vertices in such stars. Therefore, we treat such equivalence classes just as we would in the previous case. Overall, the guess of the template sequence $S$ happens over the equivalence classes determined according to Definition~\ref{def-equivalent-stars} with respect to $X_b$ and the two special equivalence classes that we just defined, namely $T^\star$ and $T^\prime$. In a valid sequence, the equivalence class $T^\prime$ must appear at least $\ell$ times. Representatives here have to be chosen from a bounded border, so a brute force approach suffices in this situation. Representatives from $T^\star$ are chosen arbitrarily and for all other cases they are chosen according to the description in Case 1.

To summarize, we now have a strategy for identifying a representative vertex from each equivalence class. For each possible defending sequence $S$,  by using Lemma~\ref{lem-stcheck}, we check whether it is valid and count the
number of vertices of saved by $S$. We output the strategy which saves the maximum number of vertices.
\end{proof}

\subsection{Parameterization by Distance to Diameter Two Graphs}

In this section, we show that \skv{} is $W[1]$-hard when parameterized by $k$ and distance to diameter two graphs 
 by giving a reduction from the $k$-clique problem. The reduction is similar, in spirit, to the one used in \cite{bazgan2014parameterized}.
 
\begin{theorem}\label{th-diameter-two}
\skv{} is $W[1]$-hard parameterized by $(k+l)$, where $l$ is the distance from the class of graphs with diameter two.
\end{theorem}
\begin{proof}

\begin{figure}[h]
 \centering
\[\begin{tikzpicture}
	\vertex (1) at (0,0) [label=left:$1$]{};
	\vertex (2) at (0,1) [label=left:$2$]{};
	\vertex (3) at (0.5,2) [label=above:$3$]{};
	\vertex (4) at (1,1) [label=right:$4$]{};
	\vertex (5) at (1,0) [label=right:$5$]{};
	\path 
		(1) edge (2)
		(2) edge (3)
		(3) edge (4)
		(4) edge (5)
		(4) edge (2)
	;
\hspace{0.5cm}
	\node[] at (4,4) {$D$};
	\draw (0+2.5,-0.5) -- (3.5+2,-0.5) -- (3.5+2,3.75) -- (0+2.5,3.75) -- (0+2.5,-0.5);
	\vertex (11) at (0+3,3) [label=above:$d_{11}$]{}; \vertex (21) at (1+3,3) [label=above:$d_{21}$]{};\vertex (31) at (2+3,3) [label=above:$d_{31}$]{};
	\vertex (12) at (0+3,2) [label=above:$d_{12}$]{}; \vertex (22) at (1+3,2) [label=above:$d_{22}$]{};\vertex (32) at (2+3,2) [label=above:$d_{32}$]{};
	\vertex (s) at (0+2,1.5) [label=above:$s$]{};
	\vertex (13) at (0+3,1) [label=above:$d_{13}$]{}; \vertex (23) at (1+3,1) [label=above:$d_{23}$]{};\vertex (33) at (2+3,1) [label=above:$d_{33}$]{};
	\vertex (14) at (0+3,0) [label=below:$d_{14}$]{};\vertex (24) at (1+3,0) [label=below:$d_{24}$]{};\vertex (34) at (2+3,0) [label=below:$d_{34}$]{};
	
	\node[] at (5+1.9,4) {$V$};
	\draw (5+2,1.5) ellipse (0.8cm and 2cm);
	\vertex (c_1) at (5+2,3) [label=above:$c_1$]{};
	\vertex (c_2) at (5+2,2.4) [label=below:$c_2$]{};
	\vertex (c_3) at (5+2,1.7) [label=below:$c_3$]{};
	\vertex (c_4) at (5+2,0.9) [label=below:$c_4$]{};
	\vertex (c_5) at (5+2,0.2) [label=below:$c_5$]{};
	\vertex (u) at (5+1.5,-0.9) [label=below:$z$]{};
	\node[] at (5+4.25,4.1) { $E$};
	\draw (5+4.3,1.6) ellipse (0.8cm and 1.8cm);
	\vertex (i_12) at (5+4,2.8) [label=right:$i_{12}$]{};
	\vertex (i_23) at (5+4,2.2) [label=right:$i_{23}$]{};
	\vertex (i_24) at (5+4,1.6) [label=right:$i_{24}$]{};
	\vertex (i_34) at (5+4,1) [label=right:$i_{34}$]{};
	\vertex (i_45) at (5+4,0.4) [label=right:$i_{45}$]{};
	\draw [thick](6.55,-0.9) -- (9,-0.05);
	\draw [thick](6.55,-0.9) -- (6.8,-0.4);
		\path 
		(s) edge (11)(s) edge (12)(s) edge (13)(s) edge (14)
		(11) edge (21)		(11) edge (22)		(11) edge (23)		(11) edge (24)		(12) edge (21)
		(12) edge (22)		(12) edge (23)		(12) edge (24)		(13) edge (21)		(13) edge (22)
		(13) edge (23)		(13) edge (24)		(14) edge (21)		(14) edge (22)		(14) edge (23)
		(14) edge (24)		(21) edge (31)		(21) edge (32)		(21) edge (33)		(21) edge (34)
		(22) edge (31)		(22) edge (32)		(22) edge (33)		(22) edge (34)		(23) edge (31)
		(23) edge (32)		(23) edge (33)		(23) edge (34)		(24) edge (31)		(24) edge (32)
		(24) edge (33)		(24) edge (34)
		(31) edge (c_1)(31) edge (c_2)(31) edge (c_3)(31) edge (c_4)(31) edge (c_5)
		(32) edge (c_1)(32) edge (c_2)(32) edge (c_3)(32) edge (c_4)(32) edge (c_5)
		(33) edge (c_1)(33) edge (c_2)(33) edge (c_3)(33) edge (c_4)(33) edge (c_5)
		(34) edge (c_1)(34) edge (c_2)(34) edge (c_3)(34) edge (c_4)(34) edge (c_5)
		(u) edge (31) (u) edge (32) (u) edge (33) (u) edge (34) 
		(c_1) edge (i_12) (c_2) edge (i_12) (c_2) edge (i_23) (c_3) edge (i_23) (c_2) edge (i_24) (c_4) edge (i_24)
		(c_3) edge (i_34) (c_4) edge (i_34) (c_4) edge (i_45) (c_5) edge (i_45)

	;
\end{tikzpicture}\]
\caption{ \hspace{0.1cm}Graph $G$ (left) and Distance to diameter two graph $G'$ (right). The vertex $z$ is adjacent to all vertices in $V \cup E$, denoted with thick edges.}
\label{fig-diameter}
\end{figure}
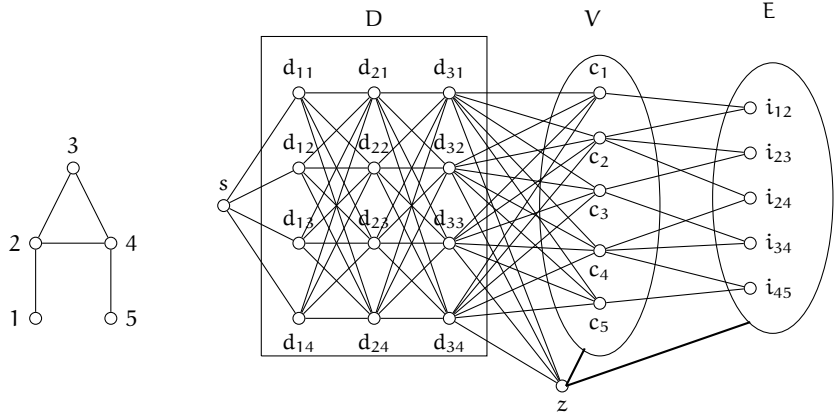

 Let $(G,k)$ be a instance of $k$-clique problem. We construct a graph $G'$ from $G$ as follows (see Fig.~\ref{fig-diameter}).
 For each vertex $v$ in $V(G)$ add vertex $c_v$ in $G'$: this set of vertices is denoted by $V$. 
 For each edge $(u,v)$ in $G$ add a vertex $i_{uv}$ in $G'$ : this set of vertices is denoted by $E$. Add an edge from vertex $i_{uv}$ in $E$ to both vertices $c_u$ and $c_v$ in $V$ for each $(u,v)$ in $E(G)$.  Add a root vertex $s$, and add vertices $d_{ij}$ (set of vertices denoted by $D$) for $1\leq i \leq k$ and $1\leq j \leq k+1$.
 Add edges between $d_{ij}$ to $d_{i'j'}$ for all $i,j,j'$ and $i'=i+1$. Connect $s$ to $d_{1,j}$ for all $j$ and connect $d_{k,j}$ to each 
 vertex in $V$ for all $j$. Add a vertex $z$ adjacent to all vertices of $V \cup E$. The vertex set of $G'$ is $V(G')=\{s,z\}\cup D \cup V\cup E$. 
 It is easy to see that $G'$ is $k(k+1)+1$ distance to diameter two graphs.
 Set $k'=k+{k \choose 2}+2$ for saving vertices in $G'$.
 
 We prove that \skv{} on $(G',s,k')$ is a yes instance if and only if $k$-clique on $(G,k)$ is a yes instance.
 Suppose that there is a clique $K$ of size $k$ in $G$. Let $S$ be a valid strategy described as follows: At each time step $t=1,\cdots,k$  
 $S$ defends the vertices $c_v$ for all $v \in K$, at time step $k+1$, $S$ defends the vertex $z$ and at time step $k+2$, $S$ defends a vertex $i_{uv}$, for some edge $(u,v) \not \in E(G[K])$.
 So the strategy $S$ saves $k+1+{k \choose 2}+1$ vertices in $G$: $k$ vertices in $V$ and ${k \choose 2}+1$ vertices in $E$.
 
 For other direction, suppose that $S=\{p_1,p_2,\cdots,p_l\}$ is a valid strategy that saves at least $k'$ vertices in $G'$.
 It is easy to see that at any time step defending a vertex $d_{ij}$ for any $i,j$ is not helpful, because there is at least one burning vertex $d_{i,j'}$ with same neighborhood as $d_{ij}$.
 Further defending vertices in $E$ does not save any vertices. Since the vertex $z$ is universal, therefore it needs to be defended. 
So the vertices in $S \cap V$ are responsible to save at least ${k \choose 2}$ many vertices, which is possible only if the vertices defended in $V$ induces $k$-clique in~$G$. \qed


\end{proof}

We can obtain the hardness of \skv{} by a reduction from the $k$-clique problem as well, in fact by making minor changes to the reduction used to prove Theorem~\ref{th-diameter-two}.

\begin{corollary}\label{coro-split}
\skv{} is $W[1]$-hard parameterized by $(k +l)$, where $l$ is distance to split graphs.
\end{corollary}
\begin{proof}
The proof is similar to Theorem~\ref{th-diameter-two}, except the following minor changes in the construction of graph $G'$.
Add an edge between every pair of vertices in $V$ and remove the universal vertex $z$. 
Add vertices $d_{ij}$ for $1\leq i \leq k-1$ and $1\leq j \leq k$.
Add edges between $d_{ij}$ to $d_{i'j'}$ for all $i,j,j'$ and $i'=i+1$. Connect $s$ to $d_{1,j}$ for all $j$ and connect $d_{k-1,j}$ to each 
vertex in $V$ for all $j$. Set $k'=k+{k \choose 2}+1$ for saving vertices in $G'$.
The graph $G'$ is $k(k-1)+1$ distance to a split graph.
It is easy to see that \skv{} on $(G',s,k')$ is a yes instance if and only if $k$-clique on $(G,k)$ is a yes instance. \qed
\end{proof}

\section{Kernelization Complexity}\label{sec-kernel}

\longversion{\subsection{Parameterization by Distance to Clique}}
In this section, we give a polynomial kernel for \skv{} when parameterized by distance to clique, as summarized in the following theorem. \shortversion{The kernelization is obtained by a combination of several observations about the behavior of the vertices in the clique, and adding a small auxiliary graph that can mimic the role of most of the clique that is not adjacent to the modulator. }

\begin{theorem}
\skv{} admits a polynomial kernel of size at most $O(l^2)$,  where $l$ is the distance to clique.
\end{theorem}

Let $(G,s,k)$ be an instance of \skv{} and $X \subseteq V(G)$ of size $l$ such that $G \setminus X=C$ is a clique. 
With out loss of generality we assume that $s \in X$, otherwise define $X'=X \cup \{s\}$ such that $G \setminus X'$ is a clique with size of modulator $l+1$.
Since the length of longest induced
path in $G$ is at most $l+1$, using Lemma~\ref{lem-optlength1} we get the following Corollary.
\begin{corollary}\label{coro-length-clique}
 Given an instance of \ff{}, then any optimal strategy can defend at most $l+1$ vertices, where $l$ is size of the clique modulator $X$. 
\end{corollary}
Let $X_L:=\{x \in X: |N(x) \cap C| \leq l+1\}$, $X_H=X \setminus X_L$. Let $J:= \{y \in C ~|~ \exists x \in X_L, y \in N(x)\}$, $|J| \leq l(l+1)$.
Our kernelization algorithm is based on the following observation. 
We show that replacing the clique $G[C \setminus J]$ with another clique of small size does
not affect the solution.

\begin{lemma}
 Let $S$ be a valid strategy containing a vertex $v \in C \setminus J$, then all the vertices of $N(v) \setminus S$ are burned.
\end{lemma}

\begin{proof}
 Since $S$ is a valid strategy, there is an induced path $P$ from $s$ to $v$ such that all vertices on $P$ are burned, except $v$.
 Let $X_v:=N(v) \cap X$ and for every $x \in X_v \setminus S$, we have $|N(x)\cap C|>l+1$: suppose there is $x \in N(v) \cap X$ such that 
 $|N(x)\cap C|\leq l+1$ then $v \in J$, contradiction to $v \in C \setminus J$.
  
 Let $u$ be a burned neighbor of $v$ on $P$. If $u \in C$ then $C \setminus S$ is burned and for every $x \in X_v \setminus S$, 
 we have $|N(x)\cap C|>l+1$. From Corollary~\ref{coro-length-clique} we can only defend at most $l+1$ vertices, therefore all 
 vertices of $X_v \setminus S$ are burned.
 
 If $u \in X_v\setminus S$ then $|N(u)\cap C|>l+1$, therefore all vertices $C \setminus S$ are burned. From the first case we can see that 
all vertices of $X_v \setminus S$ are also burned. \qed
 \end{proof}
{\bf Reduction rules}: 

\begin{enumerate}
\item Delete vertices of $C \setminus J$ from $G$.
  Add a clique $K$ of size $l+2$ and make each vertex of $K$ adjacent to all vertices in 
$X_H\cup J$.

\item Add another clique $L$ of size $min \{l+1, |C \setminus J| \}$ and for each vertex $u \in L$, add edges between $u$ and $J \cup K$.  
 \end{enumerate}
%
 
 Let $H$ be the graph obtained after applying above reduction rules. It is easy to see that
$X \subseteq V(H)$ such that $H \setminus X$ is a clique.
The size of the reduced instance $H$ is at most $l^2+4l+3$ and the reduction can be done in polynomial time. 

Let $C=G \setminus X$ and $C'=H\setminus X$. We may assume that $|C \setminus J|> 2l+3$; otherwise, trivially we get a kernel of 
size at most $l^2+4l+3$.
\begin{remark}\label{rem-strategy}
 Let $G$ be a graph and $S$ be a
valid strategy. 
If $S$ defends a subset $S_1$ of vertices in $C \setminus J$, then defending any subset $S_2$ of vertices in $L$ 
instead of $S_1$, with $|S_1|=|S_2|$ is also a valid strategy $S'$ for $H$.

Conversely let $S'$ be a strategy on $H$. If $S'$ defends a subset $S_1$ of vertices in $K \cup L$, 
then defending any subset $S_2$ of vertices in $C \setminus J$ 
instead of $S_1$, with $|S_1|=|S_2|$ is also a valid strategy $S$ for $G$.
\end{remark}
\begin{lemma}\label{lem-reduction}
 Let $G$ and $H$ be graphs as defined above and $S$, $S'$ be corresponding valid strategies as defined in the above remark. At least one vertex of $C$ is burned in $G$ by strategy $S$ iff at least one vertex of $C'$ is burned in $H$ by strategy $S'$.
\end{lemma}
\begin{proof}
 Let $u \in C'$ is burned by strategy $S'$ in $H$. If $u \in J$, then $u \in C$ will also burned by $S$ in $G$.
 If $u \in K$ and no vertex of $J$ burned by $S'$ in $H$ then, at least one vertex $x \in X_H$ gets burned and $|N(x) \cap C| >l+1$, therefore at least one vertex in 
 $N(x) \cap C$ gets burned by $S$ in $G$. 
 
 Let $u \in C$ is burned by strategy $S$ in $G$. If $u \in J$  then it will be burned by $S'$ in $H$.
 If $u \in C \setminus J$ and no vertex of $J$ burned by $S$ in $G$ then there exists burned vertex $x \in X_H$. Since $|N(x) \cap C'| >l+1$,  at least one vertex of $K$ gets burned. \qed 
\end{proof}
{\bf Claim 1} If a strategy $S$ saves at least $2l+1$ vertices in $G$ then $S$ saves entire clique.

\begin{proof}
 Let $S$ be a valid strategy which saves at least $2l+1$ vertices in $G$. Suppose assume that 
 there exists a vertex $v\in C$ burned by strategy $S$, then no vertex in $C$ is saved except the defended vertices. 
 The strategy $S$ can save at most $l-1$ vertices in $X$ and $l+1$ vertices in $C$, total $S$ saves at most $2l$ vertices which is a contradiction to the fact that $S$ saves at least $2l+1$ vertices.  \qed
\end{proof}
{\bf Claim 2} Saving one undefended vertex in $C$ is equivalent to saving entire clique.
\begin{proof}
Let $v$ be an undefended saved vertex of $C$. If any vertex of $C$ is burned by $S$ then $v$ gets burned contradicting the fact that $v$ is a saved vertex.   \qed
\end{proof}

\begin{lemma}
Let $G$ be a graph and $H$ be the graph obtained after applying reduction rules.
A strategy $S$ saves at least $k$ vertices in $G$ if and only if there exists a strategy $S'$ that saves at least $k'$ vertices in $H$.\end{lemma}
\begin{proof}
Let $S$ be any valid strategy on $G$ then we can find a corresponding strategy $S'$ on $H$ as follows.
If $S$ defends a vertex $v$ in $C \setminus J$ which is deleted in reduction procedure, then by Remark~\ref{rem-strategy} instead of $v$ we can defend 
any other non-defended vertex in $L$. Since $S$ can defend at most $l+1$ vertices in $G$ and $|L|=l+1$, we can always replace vertices of $C \setminus J$ in  $S$ by applying Remark~\ref{rem-strategy}. 

Conversely, let $S'$ be any valid strategy on $H$ then we can find its corresponding strategy $S$ on $G$ as follows. 
If $S'$ defends a vertex $v$ in $K \cup L$ which is added in reduction procedure, 
then by Remark~\ref{rem-strategy} instead of $v$ we can defend 
any other non-defended vertex in $C \setminus J$. Since $S'$ can defend at most $l+1$ vertices in $H$ and $|C \setminus J|>2l+3$, 
we can always replace vertices of $K \cup J$ in  $S'$ by applying Remark~\ref{rem-strategy}. 

The parameter $k'$ is defined as follows.
\begin{enumerate}
\item If $k \in [|C|, |C|+l-1]$, then set $k'=k-|C|+|J|+|K|+|L|\leq l^2+4l+3$.
\item If $k \in [2l+1, |C|-1]$, then set $k'=2l+1$.
\item If $k \in [1,2l]$, then set $k'=k$.
\end{enumerate}

Using Lemma~\ref{lem-reduction}, Claims 1 and 2 we can see that the strategy $S$ saves at least $k$ vertices in $G$ if and only if the strategy $S'$ saves at least $k'$ vertices in $H$. \qed


\end{proof}

\longversion{\subsection{Parameterization by Distance to Stars}}
\longversion{In this section,}\shortversion{Next,} we show that \skv{} does not admit a polynomial kernel parameterized by distance to stars unless $\NP \subseteq co\NP/poly$. \shortversion{We obtain this result by a polynomial parameter transformation from the problem of finding a maximum clique parameterized by the size of the vertex cover. We defer the details of the following theorem to the full version of the paper.}
\longversion{
First we introduce a notion of polynomial time and parameter transformation, that under 
the assumption $\NP \nsubseteq co\NP/poly$ allows us to show non-existence of polynomial kernels.
\begin{definition}\cite{bodlaender2011kernel}
 Let $P$ and $Q$ be parameterized problems, we say that there is a polynomial-parameter transformation from $P$ to $Q$, denoted $P \leq _{ppt} Q$,
 if there exists a polynomial time computable function $f: \{0,1\}^* \times \mathbb{N} \rightarrow \{0,1\}^* \times \mathbb{N}$ and
 polynomial $p: \mathbb{N} \rightarrow \mathbb{N}$, for all $x \in \{0,1\}^* $ and $k \in \mathbb{N}$, then following hold
 \begin{enumerate}
  \item $(x,k) \in P$ if and only if $(x',k')=f(x,k) \in Q$ 
  \item $k' \leq p(k)$
 \end{enumerate}
 The function $f$ is called a polynomial-parameter transformation from $P$ to $Q$.
\end{definition}
\begin{theorem} \cite{bodlaender2011kernel}
Let $P$ and $Q$ be parameterized problems, and suppose that $P^c$ and $Q^c$ are
the derived classical problems. Suppose that $P^c$
is $\NP$-complete, and $Q^c \in \NP$. Suppose
that $f$ is a polynomial time and parameter transformation from $P$ to $Q$. Then, if $Q$ has a
polynomial kernel, then $P$ has a polynomial kernel.
\end{theorem}
\begin{corollary}
 Let $P$ and $Q$ be parameterized problems whose unparameterized
versions $P^c$ and $Q^c$ are $\NP$-complete. If $P \leq_{ppt} Q$ and $P$ does not have a
polynomial kernel, then $Q$ does not have a polynomial kernel.
\end{corollary}

\problembox{\cl{} [\vc{}] }{A graph $G$, a vertex cover $X \subseteq V(G)$ and an integer $k$.}{The size $l := |X|$ of the vertex cover.}{Does there exists a clique of  size $k$ in $G$?}

\begin{theorem} \cite{bodlaender2014kernelization}
 \textnormal {\cl{} [\vc{}]} does not admit a polynomial kernel unless $\NP \subseteq co\NP/poly$.
\end{theorem}
}
\begin{figure}[h]
\centering
\includegraphics[scale=0.6,trim={5cm 20cm 2cm 3.1cm},clip]{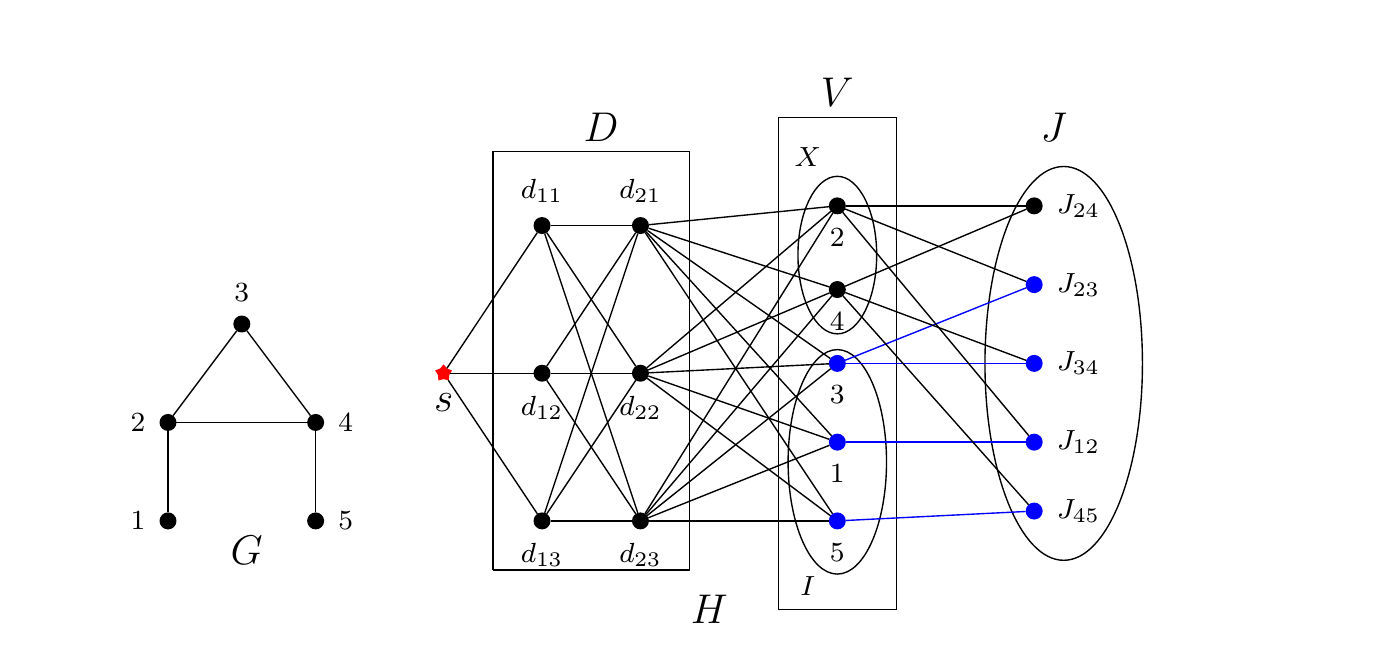}
\caption{ \hspace{0.1cm} An instance of \cl{} [\vc{}] and the constructed graph $H$ represents \skv{} [Distance to stars] in the proof of Theorem~\ref{th-nopoly} for $l=2$ and $k = 3$. In this example $X=\{2,4\}$ is 
a vertex cover of $G$ and $H \setminus (\{s\} \cup D \cup X)$ is disjoint union of stars (graph induced by blue vertices).}
\label{fig3}
\end{figure}
\begin{theorem}\label{th-nopoly}
 \textnormal \skv{} parameterized by the distance to  disjoint union of stars does not admit a polynomial kernel, unless $NP \subseteq co\NP/poly$.
\end{theorem}
\begin{proof}
 We show that there is
a polynomial-parameter transformation from \cl{} [\vc{}] to \skv{} [Distance to stars].
Let $(G, X, l, k)$ be an instance of \cl{} [\vc{}]. We construct a graph $H$ (see Fig.~\ref{fig3}) from $G$ as follows.
Initialize $H$ as a copy of $G$, this set of vertices are denoted by $V=X \cup I$
For each edge $uv \in E(G)$ add a vertex $J_{uv}$ to $H$ and make it adjacent to both vertices $u$ and $v$. we denote this set of 
vertices with $J$.

Add a root vertex $s$, and add vertices $d_{ij}$ (set of vertices denoted by $D$) for $1\leq i \leq k-1$ and $1\leq j \leq k$.
 Add edges between $d_{ij}$ to $d_{i'j'}$ for all $i,j,j'$ and $i'=i+1$. Connect $s$ to $d_{1,j}$ for all $j$ and connect $d_{k-1,j}$ to each 
 vertex of $G$ in $H$ for all $j$. Vertex set of $H$ is $V(H)=\{s\}\cup D \cup V\cup J$. 
 It is easy to see that $H \setminus(\{s\}\cup D \cup X)$ is a disjoint union of stars. 
 Therefore $H$ is $k(k-1)+l+1$ distance to stars with $X'=\{s\} \cup D \cup X$ as a modulator. 
 Since $k \leq l+1$, $H$ is at most $(l+1)^2$ distance to stars.
 Let $l':= |X'|$ is polynomially bounded in the input parameter $l$ and we can construct the graph $H$ in polynomial time. 
 Set $k'=k+{k \choose 2}+1$ for saving vertices in $H$. Let $(H, X',l',k')$ is an instance of
the \skv{} problem parameterized by a distance to stars.

We show that \skv{} on $(H, X',l',k')$ is a yes instance if and only if there is a $k$-clique on $(G, X, l, k)$ is a
yes instance. 
Suppose that $(G,X,l,k)$ is a yes-instance of \cl{} [\vc{}] and $C$ be a $k$-clique in $G$.
We can find a valid strategy $S$ which saves at least $k'$ vertices in $H$ as follows:
For first $k$ time steps $S$ defends $k$ vertices of $C$ in $H$.  At time step $k+1$, $S$ defends some vertex 
$J_{uv}$ in $H$ such that $uv \not \in E(C)$. 
The strategy $S$ saves $k$ vertices in $C$ and saves ${k \choose 2}$ vertices in $H$ corresponding to clique edges in $G$ and the 
vertex $J_{uv}$.


Conversely, suppose that 
$S$ is a valid strategy that saves at least $k'$ vertices in $H$.
It is easy to see that at any time step, defending a vertex $d_{ij}$ for any $i,j$ is not helpful, 
because there is at least one burning vertex $d_{i,j'}$ with same neighborhood as $d_{ij}$.
Further protecting $J_{uv}$ vertices for some $uv \in E(G)$ does not save any vertices.
So the vertices in $S \cap V$ are responsible for saving at least ${k \choose 2}$ many vertices.

If $|S \cap I|=d \geq 2$ then in the best case $S$ can save at most $k''=k+1+{k-d \choose 2} + d(k-d) =k+1+(k-d)(k+d-1)/2$ vertices in $H$.
Observe that $k''$ is less than $k'$ for $d \geq 2$, contradiction to fact that the strategy $S$ saves at least $k'$ vertices 
in $H$. Therefore  $|S \cap I|\leq 1$.
If $S \cap I = \emptyset$ then the strategy $S$ saves at least $k'$ vertices if and only if the vertices defended in $S \cap X$ are
forms a $k$-clique in $G$.
If $|S \cap I| = 1$ then $S$ defends one vertex $v$ in $I$ and $k-1$ vertices in $X$. 
$S$ saves at least $k'$ vertices if
$k-1$ vertices defended in $X$ induces a $(k-1)$-clique in $G$ and $v$ is adjacent to $k-1$ defended vertices.
Combining both cases,  $S$ can save at least $k'$ vertices in $H$ if the vertices defended in $S \cap V$ forms a $k$-clique in $G$.\qed

\end{proof}

\longversion{\section{Conclusion}
In this paper, we studied the parameterized complexity of \ff{} problem for various structural parameters.
We considered the size of a modulator to threshold graphs, cluster graphs and disjoint unions of stars as parameters, and showed \fpt{} algorithms in all cases. \longversion{We also established that \ff{} admits a polynomial kernel when parameterized by the size of a modulator to a clique, while it is unlikely to admit a polynomial kernel when
parameterized by the size of a modulator to a disjoint union of stars.}\shortversion{We also explored the kernelization complexity of the problem.} In contrast to the tractability results, we found that \ff{} is \WOH{} when parameterized by the distance to split graphs.

The following problems remain open. Does \skv{} admit a polynomial kernel when parameterized by $k$ and the size of a vertex cover? Also, is the \ff{} problem \fpt{} when parameterized by distance to cographs?
}

\bibliographystyle{splncs03}
\bibliography{myrefs.bib}
\newpage
\clearpage
\appendix

\end{document}